\documentclass{system-security-article}

\usepackage{authblk}

\newcommand{\floor}[1]{\left\lfloor #1 \right\rfloor}
\newcommand{\ceil}[1]{\left\lceil #1 \right\rceil}
\newcommand{\abs}[1]{\left| #1 \right|}

\newcommand{\set}[1]{\left\{ #1 \right\}}
\newcommand{\setbuild}[2]{\left\{ #1 \,:\, #2 \right\}}

\providecommand{\href}[2]{#2}

\title{Canonical Byte-String Encoding for Finite-Ring Cryptosystems}

\author[1]{Kyrylo Riabov\thanks{\texttt{kyryl.riabov@gmail.com}; \href{https://orcid.org/0009-0003-4118-8492}{ORCID 0009-0003-4118-8492}}}
\author[1]{Serhii Kryvyi\thanks{\texttt{sl.krivoi@gmail.com}; \href{https://orcid.org/0000-0003-4231-0691}{ORCID 0000-0003-4231-0691}}}

\affil[1]{Taras Shevchenko National University of Kyiv, Faculty of Computer Science and Cybernetics, 4d Akademika Glushkova Ave., Kyiv, 03680, Ukraine}

\date{}

\newcommand{\Zm}{\ensuremath{\mathbb{Z}_m}}
\newcommand{\Gm}{\ensuremath{G_m}}

\newcommand{\ByteString}{\mathsf{ByteString}}
\newcommand{\ByteStringBounded}{\ByteString_{w}}
\newcommand{\EncodeM}{\mathsf{Encode}_m}
\newcommand{\DecodeM}{\mathsf{Decode}_m}
\newcommand{\EncodePrefixM}{\mathsf{EncodePrefix}_m}
\newcommand{\DecodePrefixM}{\mathsf{DecodePrefix}_m}
\newcommand{\List}{\mathsf{List}}
\newcommand{\Fin}{\mathsf{Fin}}
\newcommand{\Except}{\mathsf{Except}}
\newcommand{\Error}{\mathsf{Error}}
\newcommand{\OK}{\mathsf{ok}}
\newcommand{\Some}{\mathsf{some}}

\newcommand{\concat}{\mathbin{\|}}
\newcommand{\BaseM}{base-\(m\)}
\newcommand{\BaseMLen}{\BaseM{}-len}
\newcommand{\BytesVar}{\mathit{bytes}}

\newcommand{\LeanParams}[2]{\href{https://github.com/KyrylR/phd-symmetric-cryptography/blob/8e903222a7df1c4365592e4b16fa752a924b8977/lean/encoding-dr-1/EncodingProofs/BaseMLen/Params.lean\#L#1-L#2}{[Params.lean, L#1--L#2]}}
\newcommand{\LeanSpec}[2]{\href{https://github.com/KyrylR/phd-symmetric-cryptography/blob/8e903222a7df1c4365592e4b16fa752a924b8977/lean/encoding-dr-1/EncodingProofs/BaseMLen/Spec.lean\#L#1-L#2}{[Spec.lean, L#1--L#2]}}
\newcommand{\LeanPrefix}[2]{\href{https://github.com/KyrylR/phd-symmetric-cryptography/blob/8e903222a7df1c4365592e4b16fa752a924b8977/lean/encoding-dr-1/EncodingProofs/BaseMLen/Prefix.lean\#L#1-L#2}{[Prefix.lean, L#1--L#2]}}
\newcommand{\LeanPayload}[2]{\href{https://github.com/KyrylR/phd-symmetric-cryptography/blob/8e903222a7df1c4365592e4b16fa752a924b8977/lean/encoding-dr-1/EncodingProofs/BaseMLen/Payload.lean\#L#1-L#2}{[Payload.lean, L#1--L#2]}}
\newcommand{\LeanRenorm}[2]{\href{https://github.com/KyrylR/phd-symmetric-cryptography/blob/8e903222a7df1c4365592e4b16fa752a924b8977/lean/encoding-dr-1/EncodingProofs/BaseMLen/Renorm.lean\#L#1-L#2}{[Renorm.lean, L#1--L#2]}}
\newcommand{\LeanStream}[2]{\href{https://github.com/KyrylR/phd-symmetric-cryptography/blob/8e903222a7df1c4365592e4b16fa752a924b8977/lean/encoding-dr-1/EncodingProofs/BaseMLen/Stream.lean\#L#1-L#2}{[Stream.lean, L#1--L#2]}}
\newcommand{\LeanExamples}[2]{\href{https://github.com/KyrylR/phd-symmetric-cryptography/blob/8e903222a7df1c4365592e4b16fa752a924b8977/lean/encoding-dr-1/EncodingProofs/BaseMLen/Examples.lean\#L#1-L#2}{[Examples.lean, L#1--L#2]}}

\begin{document}

    \maketitle

    \begin{abstract}
        Ring-mapping cryptosystems need a canonical procedure that converts ordinary byte strings
        into valid residues before any algebraic encryption step can begin.
        This paper isolates that representation layer and presents the \BaseMLen{} codec, a
        canonical map from byte strings of length less than \(2^{w}\) to lists of residues modulo \(m\).
        The construction uses two fixed-width \BaseM{} headers, for the byte length and the final
        normalized state, together with a uniform payload transduction inspired by the normalized-state
        viewpoint of ANS and rANS~\cite{duda2013ans,giesen2014rans}.

        The resulting format is deterministic, self-delimiting, and parameterized by the target modulus,
        which makes it suitable as a transport layer for finite-ring protocols rather than as an entropy coder.
        Its length-first wire format keeps the abstract specification close to the implementation interface:
        the decoder can recover the declared output size before processing the remaining payload stream,
        while the encoder remains linear in the emitted digit stream.
        For every supported modulus, decoding reconstructs the original bytes exactly; because the byte length
        is encoded explicitly, decoding also tolerates appended valid suffix digits and stops after the declared
        payload has been recovered.

        The paper links the abstract codec to a companion Rust implementation and a Lean 4 formalization
        with machine-checked proofs.
        The formal development establishes fixed-width header inversion, normalization-window bounds on the
        stored payload state, stream-level roundtrip correctness, and validity of every emitted residue.

        We also analyze asymptotic representation cost, report representative native-code throughput,
        and discuss practical limits such as exposed message length, implementation-level allocation bounds,
        and the distinction between deterministic serialization and cryptographic protection.
    \end{abstract}


    \section{Introduction}
    \label{sec:intro}

    Finite-ring cryptosystems such as the ring-image construction in~\cite{kryvyi2025symmetric}
    require messages to be represented as elements of \(\Zm\) before encryption or decryption.
    Existing examples use hand-built alphabets and small fixed symbol tables, which illustrate the algebra
    but leave open the general problem of how arbitrary UTF-8 text~\cite{rfc3629}
    or binary data should be represented before the ring-level encryption step begins.

    This representation problem extends beyond the mentioned papers.
    Any protocol that operates on plaintext data over \(\Zm\) must first represent that
    data as residues in \(\Zm\) before applying its core transformations.
    Canonical byte-level encodings are standard in other settings, including base-\(N\) transfer encodings~\cite{rfc4648},
    ASN.1 canonical and distinguished encoding rules~\cite{x690}, deterministic CBOR-based interchange~\cite{rfc8949},
    and public-key standards with explicit octet-string and integer conversion rules~\cite{rfc8017}.
    Those formats solve adjacent interoperability problems, but they do not provide a modulus-parameterized residue
    stream tailored to protocols that natively operate over \(\Zm\).
    This aspect is also reflected in the ERC project SYMPZON~\cite{sympzon}, which studies symmetric
    primitives over integer rings for applications such as format-preserving encryption,
    multi-party computation, fully homomorphic encryption, and zero-knowledge.
    The present work provides a concrete representation layer that can support such constructions.

    This paper defines a canonical byte-to-residue codec for finite-ring cryptosystems.
    The preferred \BaseMLen{} construction maps each byte string of length less than \(2^{w}\) to a residue stream in \(\Zm\).
    Its wire format consists of three components: a fixed-width base-\(m\) encoding of the byte length and the final
    payload state, and a payload stream produced by a uniform base-\(m\) transduction.
    The decoder reconstructs the original bytes exactly and stops after the declared length,
    so any additional valid suffix digits are ignored.

    The present construction is not an entropy coder.
    Its asymptotic payload cost reflects a change of output alphabet rather than a reduction in information content;
    larger moduli may reduce the number of transmitted residue symbols, but bit-level compression would require an
    explicit probabilistic source model and lies outside the scope of this paper.
    The payload mechanism follows the normalized-state perspective of ANS and rANS, but is used here to
    obtain a deterministic, self-delimiting residue format rather than a new entropy coder or security primitive~\cite{duda2013ans,giesen2014rans}.
    The problem addressed is orthogonal to ring-based lattice cryptography such as
    NTRU and ML-KEM~\cite{ntru1998,nist2024post}, which rely on different assumptions and target different protocol goals.

    The paper makes three contributions.
    \begin{enumerate}
        \item It defines a canonical, modulus-parameterized byte-to-residue codec for finite-ring protocols,
        centered on the \BaseMLen{} wire format.
        \item It establishes exact decoding, suffix tolerance, and fixed-width header inversion for supported
        moduli and byte strings of length less than \(2^{w}\).
        \item It provides a reference implementation and formal specification, consisting of a
        Rust library and a Lean 4 formalization with machine-checked proofs.
    \end{enumerate}

    The remainder of the paper formalizes the model and proof scope, specifies the codec, states the
    main guarantees, and evaluates its complexity and practical performance.

    \section{Model, implementation scope, and Lean 4 formalization}
    \label{sec:model}

    \subsection{Core definitions and boundaries}
    \label{subsec:assumptions}

    Throughout the abstract development, we fix a global word width \(w \in \mathbb{N}\) and study those moduli \(m\)
    for which the codec parameters are well formed at that width.
    The evaluation in \Cref{sec:evaluation} later instantiates this generic theory at \(w=64\).

    \begin{definition}[Supported modulus]
        \label{def:supported-modulus}
        For the fixed word width \(w\), a \emph{supported modulus} is a modulus \(m\)
        for which the abstract codec parameters are well formed, namely
        \[
            2 \le m
            \qquad \text{and} \qquad
            256 \le \mathsf{decoderLowerBound}(w,m).
        \]
        This definition is formalized in \LeanParams{51}{55}.
    \end{definition}

    \begin{definition}[Admissible byte string]
        \label{def:byte-string-bounded}
        An \emph{admissible byte string} is a finite sequence of bytes of length less than \(2^{w}\).
        We write
        \[
            \ByteStringBounded =
            \setbuild{\BytesVar \in \set{0,1,\ldots,255}^*}{\abs{\BytesVar} < 2^{w}}.
        \]
        This definition is formalized in \LeanSpec{17}{19}.
    \end{definition}

    For each supported modulus \(m\), the abstract codec interface is
    \[
        \EncodeM : \ByteStringBounded \to \List(\Fin(m)),
        \qquad
        \DecodeM : \List(\Fin(m)) \to \Except(\Error, \ByteStringBounded).
    \]
    The encoder produces valid \BaseM{} digits only, and the decoder either reconstructs an
    admissible byte string or reports failure.
    The abstract encoder and decoder are formalized in \LeanSpec{146}{182},
    while the fact that every emitted symbol lies in \([0,m)\) is recorded in \LeanStream{27}{33}.

    The core model treats residue digits as elements of \(\Fin(m)\) from the outset, so malformed raw
    naturals outside the residue range do not appear inside the proved interface; see \LeanParams{32}{34}
    for the digit type and \LeanSpec{184}{239} for the boundary validators on plain numeric inputs.
    The zero-length case is explicit as well: if the length prefix decodes to \(0\), the
    decoder returns the empty byte string immediately, as in \LeanSpec{159}{169}.
    The corresponding zero-length decoding lemmas appear in \LeanStream{135}{172}.

    \subsection{Implementation and proof scope}
    \label{subsec:boundary}

    The public artifacts consist of a Rust library and a Lean 4 formalization with machine-checked proofs.
    Both are released in the same project repository~\cite{riabov2026encodingrepo}.
    On the Rust side, the main byte-oriented entry points are \(\mathsf{encode\_bytes\_base\_m\_len}\) and
    \(\mathsf{decode\_bytes\_base\_m\_len}\), with UTF-8 frontends above them.
    On the proof side, Lean formalizes the abstract codec, including supported parameters, fixed-width prefixes,
    payload transduction, and end-to-end stream properties in Lean 4~\cite{demoura2021lean4};
    the trusted-kernel and elaboration lineage behind the environment is described in~\cite{demoura2015lean}.
    The current development is width-parametric rather than tied to one machine word size, and the repository
    includes concrete \(64\)-, \(128\)-, and \(256\)-bit instantiations and smoke theorems in \LeanExamples{58}{70}.

    The proof scope is intentionally limited to the abstract specification boundary.
    It does \emph{not} verify the Rust implementation line by line.
    Parser hardening, allocation limits, low-level I/O, and other engineering details therefore
    remain implementation concerns rather than machine-checked theorems.

    \section{The Base-m-len codec}
    \label{sec:codec}

    \subsection{Fixed-width base-m headers}
    \label{subsec:prefix}

    \begin{definition}[Prefix width]
        \label{def:prefix-width}
        For the fixed word width \(w\) and a supported modulus \(m\), let
        \[
            k = \min \set{j \in \mathbb{N}_0 : m^j \ge 2^{w}}.
        \]
        The fixed-width prefix encoder \(\EncodePrefixM\) represents every \(x < 2^{w}\)
        as exactly \(k\) little-endian \BaseM{} digits.
    \end{definition}

    The same width \(k\) is used for both fixed prefixes in the wire format: the byte length and the final payload state.
    In the Lean 4 formalization, the definition of this width and its least-width property appear in \LeanParams{123}{157},
    with bundled forms for parameter packages in \LeanParams{215}{271}.

    \subsection{Uniform base-m payload transduction}
    \label{subsec:payload}

    Let \(L\) denote the decoder lower bound associated with the supported modulus.
    The payload layer maintains a \BaseM{} state and applies the uniform byte update
    \[
        x \leftarrow 256x + b.
    \]
    In the Lean 4 formalization, the bundled definitions of the decoder lower bound \(L\) and the encoder
    threshold appear in \LeanParams{220}{231}, and the encoder-side and decoder-side renormalization procedures
    together with the payload transducer are defined in \LeanSpec{81}{144}.
    Operationally, the encoder scans the input bytes from right to left.
    Before each byte update, it emits \BaseM{} digits until the current state is below the encoder threshold.
    After the full reverse scan, the emitted payload digits are exposed in FIFO order,
    so the decoder can read the final wire format from left to right.

    This mechanism is close in spirit to ANS and rANS, but the purpose is different.
    Classical ANS work studies entropy coding relative to a probabilistic source model.
    Here the byte update is uniform, the output alphabet is fixed by the target modulus \(m\),
    and the goal is a deterministic, self-delimiting residue representation that can serve as a
    canonical message layer for finite-ring protocols~\cite{duda2013ans,giesen2014rans}.

    \subsection{Wire format and algorithms}
    \label{subsec:wire}

    \begin{definition}[Wire format]
        \label{def:wire-format}
        For every admissible byte string \(\BytesVar\),
        \[
            \EncodeM(\BytesVar) = \mathit{len} \concat \mathit{state} \concat \mathit{payload},
        \]
        where \(\mathit{len} = \EncodePrefixM(\abs{\BytesVar})\), \(\mathit{state}\) is the fixed-width encoding of
        the final payload state, and \(\mathit{payload}\) is the FIFO \BaseM{} digit stream produced by the payload transducer.
    \end{definition}
    In the Lean 4 formalization, this concatenated wire format is the definition of the abstract encoder; see \LeanSpec{146}{153}.

    Decoding is staged in the same order.
    The decoder first reads the length prefix.
    If that prefix is \(0\), it returns the empty byte string.
    Otherwise, it reads the state prefix and then consumes exactly as many payload digits
    as are needed to reconstruct the declared number of bytes.
    Any valid digits that remain after those bytes are recovered are left unread.
    In the Lean 4 formalization, the staged decoder is defined in \LeanSpec{159}{182}; its immediate zero-length return
    already appears in \LeanSpec{159}{169}, with dedicated zero-length stream lemmas in \LeanStream{135}{172}.
    The padding-tolerant end-to-end theorem that trailing suffix digits are ignored appears in \LeanStream{174}{266}.

    \begin{algorithm}[t]
    \caption{Base-m-len encode}
    \KwInput{Admissible byte string \(\BytesVar\), supported modulus \(m\)}
    \KwOutput{Residues in \(\List(\Fin(m))\)}
    Set \(k \gets \min \set{j \in \mathbb{N}_0 : m^j \ge 2^{w}}\)\;
    Encode \(\abs{\BytesVar}\) as \(k\) little-endian \BaseM{} digits\;
    Initialize the payload state \(x\) at the decoder lower bound \(L\)\;
    \For{byte values \(b\) of \(\BytesVar\) in reverse order}{
        While \(x\) exceeds the encoder threshold, emit \(x \bmod m\) and set
        \(x \gets \floor{x/m}\)\;
        Set \(x \gets 256x + b\)\;
    }
    Encode the final state \(x\) as \(k\) little-endian \BaseM{} digits\;
    Reverse the emitted payload digits and return
    \(\mathit{len} \concat \mathit{state} \concat \mathit{payload}\)\;\label{alg:encode-alg}
    \end{algorithm}

    \begin{algorithm}[t]
    \caption{Base-m-len decode}
    \KwInput{Residue stream \(digits\), supported modulus \(m\)}
    \KwOutput{\(\Except(\Error, \ByteStringBounded)\)}
    Set \(k \gets \min \set{j \in \mathbb{N}_0 : m^j \ge 2^{w}}\)\;
    Read the first \(k\) digits as the declared byte length \(\ell\)\;
    \If{\(\ell = 0\)}{
        Return \(\OK([])\)\;
    }
    Read the next \(k\) digits as the initial payload state \(x\)\;
    \For{\(i = 1\) \KwTo \(\ell\)}{
        Output \(x \bmod 256\)\;
        Set \(x \gets \floor{x/256}\)\;
        While \(x\) is below the decoder lower bound, consume the next payload digit
        \(d\) and set \(x \gets xm + d\)\;
    }
    Ignore any remaining trailing suffix and return the reconstructed bytes\;
    \label{alg:decode-alg}
    \end{algorithm}

    In the Lean 4 formalization, Algorithms~\ref{alg:encode-alg} and~\ref{alg:decode-alg} correspond directly to the definitions in \LeanSpec{81}{182}.

    \begin{example}[Worked example for ``Hi'' at \(m=50\) and \(w=64\)]
        In the Lean 4 formalization, the companion examples in \LeanExamples{90}{93} evaluate the admissible byte string \([72,105]\) to the residue stream
        \begin{lstlisting}[numbers=none,label={lst:lstlisting}]
[2, 0, 0, 0, 0, 0, 0, 0, 0, 0, 0, 0, 12, 8, 11, 36, 6, 32, 19, 0, 38, 1, 49, 1, 1, 48]
        \end{lstlisting}
        Here \(k = 12\).
        The first 12 digits are the fixed-width encoding of the byte length
        \(2\), the next 12 digits encode the stored payload state, and the final two digits form the payload.
        Decoding returns \(\OK([72,105])\), hence the original two-byte input.
    \end{example}

    \section{Correctness guarantees}
    \label{sec:correctness}

    The proof development is split across six Lean files, from parameters and fixed-width prefixes to
    renormalization, payload inversion, and stream-level correctness.

    \subsection{Main theorems}
    \label{subsec:main-theorems}

    \begin{theorem}[Stream correctness]
        \label{thm:stream-roundtrip}
        For every supported modulus \(m\) and every admissible byte string
        \(\BytesVar \in \ByteStringBounded\),
        \[
            \DecodeM(\EncodeM(\BytesVar)) = \OK(\BytesVar).
        \]
        Moreover, for every valid suffix \(t \in \List(\Fin(m))\),
        \[
            \DecodeM(\EncodeM(\BytesVar) \concat t) = \OK(\BytesVar).
        \]
    \end{theorem}

    \begin{proof}
        In the Lean 4 formalization, the suffix-tolerant theorem is \LeanStream{174}{254}.
        Its empty-suffix specialization is \LeanStream{256}{266}.
        The proof uses the header-parsing lemmas in \LeanStream{35}{133}.
        The zero-length lemmas appear in \LeanStream{135}{172}, and the payload roundtrip theorem in
        \LeanPayload{79}{141}.
        It first recovers the declared length prefix, then the stored state prefix, and finally applies the
        payload roundtrip theorem to the remaining suffix.
        The extra-suffix clause follows because decoding stops after reconstructing the declared number
        of bytes and therefore leaves any remaining valid suffix unread.
    \end{proof}

    \begin{theorem}[Header correctness]
        \label{thm:k-minimal}
        For every supported modulus \(m\), let
        \[
            k = \min \set{j \in \mathbb{N}_0 : m^j \ge 2^{w}}.
        \]
        Then
        \[
            m^k \ge 2^{w}
            \qquad \text{and} \qquad
            \forall j < k,\; m^j < 2^{w}.
        \]
        Moreover, for every \(x < 2^{w}\),
        \[
            \DecodePrefixM(\EncodePrefixM(x)) = \Some(x).
        \]
    \end{theorem}

    \begin{proof}
        In the Lean 4 formalization, the least-width property is formalized in \LeanParams{123}{157},
        with the bundled statement in \LeanParams{260}{271}; fixed-width prefix inversion is proved in \LeanPrefix{65}{89}.
        Minimality is immediate from the definition of \(k\) as the least exponent with \(m^k \ge 2^{w}\).
        Once that bound is known, every \(x < 2^{w}\) lies in the domain of the fixed-width \BaseM{} prefix decoder.
    \end{proof}

    \begin{proposition}[Payload-state bounds]
        \label{prop:state-window}
        For every supported modulus \(m\) and every admissible byte string \(\BytesVar\), the payload encoder
        returns a state in the normalization window
        \[
            L \le \mathsf{state}(\BytesVar) < Lm.
        \]
        In particular,
        \[
            \mathsf{state}(\BytesVar) < 2^{w}.
        \]
    \end{proposition}

    \begin{proof}
        In the Lean 4 formalization, the normalization-window bound and its \(2^{w}\)
        corollary are proved in \LeanPayload{17}{77}.
        The first part is the inductive invariant of the payload transducer; the second follows from the
        global estimate \(Lm \le 2^{w}-1\).
    \end{proof}

    \subsection{Proof idea and traceability}
    \label{subsec:proof-skeleton}

    The end-to-end proof has three layers.
    In the Lean 4 formalization, fixed-width prefix inversion appears in \LeanPrefix{30}{89},
    renormalization roundtrip in \LeanRenorm{219}{294}, payload roundtrip in \LeanPayload{79}{141},
    and the final stream theorem in \LeanStream{174}{266}, with the zero-length branch in \LeanStream{135}{172}.
    Fixed-width \BaseM{} headers invert because the chosen width covers the bounded \(w\)-bit range;
    renormalization blocks roundtrip the state while preserving untouched suffixes; and a byte-level induction
    lifts that result to the payload and then to the full stream theorem.

    \section{Evaluation}
    \label{sec:evaluation}

    The evaluation addresses two questions: what the representation costs in theory,
    and whether the preferred streaming construction is practical in native code.
    Throughout this section, we instantiate the general theory at \(w=64\), so the reported header widths and
    thresholds are the concrete \(64\)-bit values of the generic formulas from the preceding sections.

    \subsection{Complexity and representation cost}
    \label{subsec:expansion}

    Let \(n = \abs{\BytesVar}\) and let \(k = \ceil{\log_m 2^{64}}\).
    The two fixed-width headers contribute exactly \(2k\) digits.
    The payload contributes a number of digits proportional to the byte length, with asymptotic rate
    \[
        \log_m 256 = \frac{8}{\log_2 m}
    \]
    digits per byte.
    \Cref{fig:expansion-curves} visualizes this trade-off over the range \(25 \le m \le 257\):
    larger moduli reduce the payload rate smoothly, while the fixed header width drops only when
    \(k\) crosses an integer threshold.
    Since each emitted digit is produced once and each consumed digit is read once, encoding
    and decoding are linear in the stream length.
    For every fixed supported modulus \(m\), this yields \(O(n)\) time and \(O(1)\) working state
    beyond the input and output buffers.

    \begin{figure}[t]
        \centering
        \caption{Representation cost as a function of the modulus \(m\) over the range \(25 \le m \le 257\) in the \(w=64\) instance.
        The upper panel shows the smooth payload rate \(\log_m 256\), while the lower panel shows the stepwise fixed header cost \(2k\).}
        \label{fig:expansion-curves}
        \begin{tikzpicture}
            \def\W{11.0}
            \def\Htop{3.2}
            \def\Hbot{3.0}
            \def\Gap{1.1}
            \def\Mmin{25}
            \def\Mmax{257}
            \def\RateMin{0.9}
            \def\RateMax{1.8}
            \def\HeadMin{16}
            \def\HeadMax{28}
            \pgfmathsetmacro{\TopOffset}{\Hbot + \Gap}

            \foreach \m/\label in {25/25,50/50,100/100,150/150,200/200,257/257} {
                \pgfmathsetmacro{\x}{\W * (\m - \Mmin) / (\Mmax - \Mmin)}
                \draw[gray!20] (\x, 0) -- (\x, \Hbot);
                \draw[gray!20] (\x, \TopOffset) -- (\x, \TopOffset + \Htop);
                \node[font=\scriptsize, below] at (\x, 0) {\label};
            }

            \foreach \r in {1.0,1.2,1.4,1.6,1.8} {
                \pgfmathsetmacro{\y}{\TopOffset + \Htop * (\r - \RateMin) / (\RateMax - \RateMin)}
                \draw[gray!20] (0, \y) -- (\W, \y);
                \node[font=\scriptsize, left] at (0, \y) {\pgfmathprintnumber[fixed,precision=1]{\r}};
            }

            \foreach \h in {16,18,20,22,24,26,28} {
                \pgfmathsetmacro{\y}{\Hbot * (\h - \HeadMin) / (\HeadMax - \HeadMin)}
                \draw[gray!20] (0, \y) -- (\W, \y);
                \node[font=\scriptsize, left] at (0, \y) {\h};
            }

            \draw[black!70] (0, 0) rectangle (\W, \Hbot);
            \draw[black!70] (0, \TopOffset) rectangle (\W, \TopOffset + \Htop);

            \draw[very thick, blue!70!black, samples=120, domain=\Mmin:\Mmax, smooth]
                plot ({\W * (\x - \Mmin) / (\Mmax - \Mmin)},
                      {\TopOffset + \Htop * ((ln(256) / ln(\x)) - \RateMin) / (\RateMax - \RateMin)});

            \foreach \a/\b/\d in {25/30/28,31/40/26,41/56/24,57/84/22,85/138/20,139/255/18,256/257/16} {
                \pgfmathsetmacro{\xa}{\W * (\a - \Mmin) / (\Mmax - \Mmin)}
                \pgfmathsetmacro{\xb}{\W * (\b - \Mmin) / (\Mmax - \Mmin)}
                \pgfmathsetmacro{\y}{\Hbot * (\d - \HeadMin) / (\HeadMax - \HeadMin)}
                \draw[very thick, orange!80!black] (\xa, \y) -- (\xb, \y);
            }

            \foreach \m/\rate/\hdr/\anchor/\dx/\dy in {
                25/1.7227/28/west/0.10/0.18,
                50/1.4175/24/west/0.10/0.18,
                257/0.9993/16/east/-0.10/0.18
            } {
                \pgfmathsetmacro{\x}{\W * (\m - \Mmin) / (\Mmax - \Mmin)}
                \pgfmathsetmacro{\ytop}{\TopOffset + \Htop * (\rate - \RateMin) / (\RateMax - \RateMin)}
                \pgfmathsetmacro{\ybot}{\Hbot * (\hdr - \HeadMin) / (\HeadMax - \HeadMin)}
                \fill[blue!70!black] (\x, \ytop) circle (1.5pt);
                \fill[orange!80!black] (\x, \ybot) circle (1.5pt);
                \node[font=\scriptsize, anchor=\anchor] at (\x + \dx, \ytop + \dy) {$m=\m$};
            }

            \node[font=\scriptsize\bfseries, anchor=west] at (0, \TopOffset + \Htop + 0.18) {Payload rate};
            \node[font=\scriptsize\bfseries, anchor=west] at (0, \Hbot + 0.18) {Fixed header width};
            \node[font=\scriptsize, rotate=90] at (-0.9, \TopOffset + 0.5 * \Htop) {$\log_m 256$};
            \node[font=\scriptsize, rotate=90] at (-0.9, 0.5 * \Hbot) {$2k$};
            \node[font=\scriptsize] at (0.5 * \W, -0.55) {modulus \(m\)};
        \end{tikzpicture}
    \end{figure}

    \begin{table}[t]
        \centering
        \caption{Asymptotic payload cost and fixed header size for representative moduli in the \(w=64\) instance.}
        \label{tab:expansion}
        \begin{tabular}{@{}ccc@{}}
            \toprule
            Modulus \(m\) & Payload digits per byte & Header digits \(2k\) \\
            \midrule
            25            & 1.723                   & 28                   \\
            50            & 1.417                   & 24                   \\
            257           & 0.9993                  & 16                   \\
            \bottomrule
        \end{tabular}
    \end{table}

    \subsection{Native performance}
    \label{subsec:throughput}

    The companion Rust code includes Criterion benchmarks for the preferred codec,
    exact $\mathsf{BigUint}$ radix-conversion baselines, and an external ANS reference used only for context.
    The representative measurements below were recorded on an Apple M3 Pro machine using the existing harness;
    methodology details are summarized in \Cref{sec:bench}.
    Those numbers are included only as practical evidence.

    \begin{table}[t]
        \centering
        \caption{Representative throughput measurements in MiB/s. The $\mathsf{BigUint}$
        baseline is only defined for \(m \le 256\). The ANS reference column provides context only;
        it does not implement the same length-delimited residue wire format.}
        \label{tab:benchmarks}
        \begin{tabular}{@{}lcccc@{}}
            \toprule
            Workload       & Native \(m=65\) & BigUint \(m=65\) & Native \(m=257\) & ANS ref. \\
            \midrule
            Encode 1\,KiB  & 163.6           & 25.2             & 217.0            & 218.1    \\
            Decode 1\,KiB  & 485.9           & 125.7            & 640.5            & 150.9    \\
            Encode 64\,KiB & 188.1           & 1.285            & 236.7            & 226.2    \\
            Decode 64\,KiB & 510.3           & 3.201            & 669.1            & 144.2    \\
            \bottomrule
        \end{tabular}
    \end{table}

    \subsection{Interpretation}
    \label{subsec:interpretation}

    Three conclusions are relevant for protocol design.
    \begin{enumerate}
        \item Small moduli remain usable, but \Cref{tab:expansion} makes their header and payload overhead
        explicit instead of hiding it behind ad hoc symbol tables.
        \item The preferred streaming construction is not only asymptotically linear but also much faster in
        practice than exact whole-payload radix conversion on large inputs, which is the meaningful baseline
        for a canonical ring-level representation.
        \item The \(m=257\) configuration is close to one residue per byte while still retaining explicit
        length recovery and suffix-tolerant decoding.
    \end{enumerate}

    \section{Security role and implementation limits}
    \label{sec:limits}

    This codec should be read as deterministic serialization, not as a security primitive in its own right~\cite{shannon1949}.
    Four limitations matter in practice.
    \begin{enumerate}
        \item The codec is public and deterministic, so it provides no confidentiality,
        integrity, or authenticity before the surrounding cryptographic protocol is applied.
        \item The preferred wire format exposes the exact byte length.
        If length hiding is required, padding must be enforced at the surrounding protocol layer.
        \item The Lean formalization proves the abstract specification boundary rather than the Rust
        implementation, so low-level parser and allocator behavior remain implementation responsibilities.
        \item Because decoding allocates according to the embedded length, implementations should
        enforce a caller-specified maximum decoded length for untrusted inputs.
    \end{enumerate}

    \section{Related work}
    \label{sec:related}

    Two lines of prior work are most directly relevant.

    First, ring-mapping cryptosystems motivate communication over residue domains such as
    \(\Zm\) and \(\Gm\).
    The ring-mapping cryptosystems studied in~\cite{kryvyi2022symmetric,kryvyi2025symmetric}
    provide the application setting in which byte strings must be transported through finite-ring representations.
    However, those works treat the representation layer as supporting machinery rather than as a separate codec problem.

    Second, the payload mechanism used here is structurally close to ANS and rANS~\cite{duda2013ans,giesen2014rans}.
    These methods maintain a normalized state and use renormalization while reading or emitting symbols in a fixed radix.
    The present construction adopts that operational pattern in a uniform byte setting, but uses it for deterministic,
    self-delimiting transport into a prescribed residue alphabet rather than for entropy coding relative to a source model.

    Ring-based lattice systems such as NTRU and ML-KEM~\cite{ntru1998,nist2024post} are related only at a broader structural level.
    They also work over ring-based algebraic domains, but their main concern is hardness-based encryption and key encapsulation,
    not canonical byte representation for residue-domain transport.

    Accordingly, the specific focus of this paper is the byte-to-residue representation layer itself:
    a fixed-width, self-delimiting codec with explicit correctness theorems and machine-checked proof traceability.

    \section{Conclusion}
    \label{sec:conclusion}

    The \BaseMLen{} codec isolates the byte-to-residue representation layer needed by finite-ring protocols.
    Its wire format consists of a fixed-width length prefix, a fixed-width state prefix, and a uniform \BaseM{} payload stream.
    For every fixed word width \(w\), every supported modulus, and every admissible byte string of length less than \(2^{w}\),
    decoding recovers the original bytes and is tolerant of any valid trailing suffix.

    The released Rust library~\cite{riabov2026encodingcrate} and the benchmark suite in the
    project repository indicate that the construction is practical,
    while the Lean 4 formalization in the repository~\cite{riabov2026encodingrepo}
    establishes fixed-width prefix inversion, payload-state bounds, and end-to-end stream correctness.
    Because the codec outputs canonical residue streams in \(\Zm\), it composes directly
    with the \(\Zm \leftrightarrow \Gm\) transport layer used in prior ring-mapping
    cryptosystems~\cite{kryvyi2022symmetric,kryvyi2025symmetric}.
    More broadly, it contributes a concrete message-representation layer to current work on
    symmetric cryptography over integer rings, as reflected in the SYMPZON agenda~\cite{sympzon}.

    \appendix

    \section{Benchmark methodology}
    \label{sec:bench}

    The Criterion benchmark suite in the project repository~\cite{riabov2026encodingrepo}
    is organized into four families.
    \begin{enumerate}
        \item Native byte-transduction benchmarks for
        \(\mathsf{encode\_bytes\_base\_m\_len}\) and
        \(\mathsf{decode\_bytes\_base\_m\_len}\), evaluated at moduli
        \(2,3,13,65,251,257\) and input sizes 32\,B, 1\,KiB, and 64\,KiB\@.
        \item UTF-8 smoke benchmarks for the text frontend \(\mathsf{encode\_text\_base\_m\_len}\)
        and \(\mathsf{decode\_text\_base\_m\_len}\) at modulus \(m=65\),
        using mixed UTF-8 inputs of size 1\,KiB and 64\,KiB\@.
        \item Exact \(\mathsf{BigUint}\) radix-conversion baselines for byte encode/decode at moduli \(2,3,13,65,251\)
        and the same three byte sizes; modulus \(257\) is excluded because the underlying radix-conversion
        routine supports moduli only up to \(256\).
        \item Reference ANS encode/decode benchmarks, implemented via the external \(\mathsf{constriction}\)
        library, at input sizes 32\,B, 1\,KiB, and 64\,KiB\@.
    \end{enumerate}

    Unless noted otherwise, the representative throughput values reported in \Cref{tab:benchmarks} were obtained
    on an Apple M3 Pro machine using short Criterion runs with sample size \(10\), warm-up time \(0.1\)\,s,
    and measurement time \(0.2\)\,s.

    \printbibliography[title={References}]

@article{duda2013ans,
    author       = {Duda, Jarek},
    title        = {Asymmetric Numeral Systems: Entropy Coding Combining Speed of {Huffman} Coding with Compression Rate of Arithmetic Coding},
    journal      = {arXiv preprint arXiv:1311.2540},
    year         = {2014},
    doi          = {10.48550/arXiv.1311.2540},
    eprint       = {1311.2540},
    archivePrefix= {arXiv},
    primaryClass = {cs.IT},
    url          = {https://arxiv.org/abs/1311.2540v2}
}

@online{sympzon,
    author  = {{European Commission}},
    title   = {{SYMPZON: Getting SYMmetric CryPtography Out of its Comfort ZONe}},
    year    = {2025},
    doi     = {10.3030/101160608},
    url     = {https://cordis.europa.eu/project/id/101160608},
    note    = {ERC Starting Grant, Horizon Europe, grant agreement No. 101160608, 2025--2030}
}

@software{riabov2026encodingcrate,
    author       = {Riabov, Kyrylo},
    title        = {sym-adv-encoding},
    year         = {2026},
    url          = {https://crates.io/crates/sym-adv-encoding},
    note         = {Rust crate}
}

@online{riabov2026encodingrepo,
    author       = {Riabov, Kyrylo},
    title        = {phd-symmetric-cryptography: Lean formalization of the base-n-len codec},
    year         = {2026},
    url          = {https://github.com/KyrylR/phd-symmetric-cryptography/blob/8e903222a7df1c4365592e4b16fa752a924b8977/lean/encoding-dr-1},
    note         = {GitHub repository}
}

@misc{giesen2014rans,
    author       = {Giesen, Fabian},
    title        = {{rANS} Notes},
    year         = {2014},
    url          = {https://fgiesen.wordpress.com/2014/02/02/rans-notes/},
    note         = {Practitioner notes on streaming rANS}
}

@article{shannon1949,
    author       = {Shannon, Claude E.},
    title        = {Communication Theory of Secrecy Systems},
    journal      = {Bell System Technical Journal},
    volume       = {28},
    number       = {4},
    year         = {1949},
    pages        = {656--715},
    doi          = {10.1002/j.1538-7305.1949.tb00928.x},
    url          = {https://pages.cs.wisc.edu/~rist/642-spring-2014/shannon-secrecy.pdf}
}

@misc{nist2024post,
    author       = {{National Institute of Standards and Technology}},
    title        = {Post-Quantum Cryptography Standards: {FIPS} 203 ({ML-KEM}), {FIPS} 204 ({ML-DSA}), {FIPS} 205 ({SLH-DSA})},
    year         = {2024},
    month        = {8},
    note         = {NIST post-quantum cryptography standards overview}
}

@incollection{ntru1998,
    author       = {Hoffstein, Jeffrey and Pipher, Jill and Silverman, Joseph H.},
    title        = {{NTRU}: A Ring-Based Public Key Cryptosystem},
    booktitle    = {Algorithmic Number Theory},
    series       = {Lecture Notes in Computer Science},
    volume       = {1423},
    publisher    = {Springer},
    year         = {1998},
    pages        = {267--288},
    doi          = {10.1007/BFb0054868},
    url          = {https://www.ntru.org/f/hps98.pdf}
}

@inproceedings{kryvyi2025symmetric,
    author       = {Kryvyi, Serhiy and Riabov, Kyrylo},
    title        = {Symmetric Cryptosystem Based on Ring Images},
    booktitle    = {Proceedings of the Workshop on Intelligent Information Technologies ({UkrProg-IIT} 2025)},
    series       = {CEUR Workshop Proceedings},
    volume       = {4049},
    publisher    = {CEUR-WS.org},
    year         = {2025},
    url          = {https://ceur-ws.org/Vol-4049/paper1.pdf}
}

@article{kryvyi2022symmetric,
    author       = {Kryvyi, S. L. and Opanasenko, V. N. and Grinenko, E. A. and Nortman, Yu. A.},
    title        = {Symmetric Information Exchange System Based on Ring Isomorphism},
    journal      = {Cybernetics and Systems Analysis},
    volume       = {58},
    number       = {5},
    year         = {2022},
    url          = {https://www.springerprofessional.de/en/symmetric-information-exchange-system-based-on-ring-isomorphism/23785398}
}

@techreport{rfc3629,
    author       = {Yergeau, Fran\c{c}ois},
    title        = {{UTF-8}, a Transformation Format of {ISO} 10646},
    institution  = {{RFC Editor}},
    type         = {RFC},
    number       = {3629},
    year         = {2003},
    month        = {11},
    doi          = {10.17487/RFC3629},
    url          = {https://www.rfc-editor.org/info/rfc3629}
}

@techreport{rfc4648,
    author       = {Josefsson, Simon},
    title        = {The Base16, Base32, and Base64 Data Encodings},
    institution  = {{RFC Editor}},
    type         = {RFC},
    number       = {4648},
    year         = {2006},
    month        = {10},
    doi          = {10.17487/RFC4648},
    url          = {https://www.rfc-editor.org/info/rfc4648}
}

@techreport{rfc8017,
    author       = {Moriarty, Kathleen and Kaliski, Burt and Jonsson, Jakob and Rusch, Andrew},
    title        = {{PKCS} \#1: {RSA} Cryptography Specifications Version 2.2},
    institution  = {{RFC Editor}},
    type         = {RFC},
    number       = {8017},
    year         = {2016},
    month        = {11},
    doi          = {10.17487/RFC8017},
    url          = {https://www.rfc-editor.org/info/rfc8017}
}

@techreport{rfc8949,
    author       = {Bormann, Carsten and Hoffman, Paul},
    title        = {Concise Binary Object Representation ({CBOR})},
    institution  = {{RFC Editor}},
    type         = {RFC},
    number       = {8949},
    year         = {2020},
    month        = {12},
    doi          = {10.17487/RFC8949},
    url          = {https://www.rfc-editor.org/info/rfc8949},
    note         = {STD 94}
}

@online{x690,
    author       = {{International Telecommunication Union}},
    title        = {{X.690}: Information Technology -- {ASN}.1 Encoding Rules: Specification of Basic Encoding Rules ({BER}), Canonical Encoding Rules ({CER}) and Distinguished Encoding Rules ({DER})},
    year         = {2021},
    month        = {2},
    url          = {https://www.itu.int/ITU-T/recommendations/rec.aspx?rec=14779&lang=en},
    note         = {Recommendation X.690 (02/2021) | ISO/IEC 8825-1:2021}
}

@inproceedings{demoura2015lean,
    author       = {de Moura, Leonardo and Kong, Soonho and Avigad, Jeremy and van Doorn, Floris and von Raumer, Jakob},
    title        = {The Lean Theorem Prover (System Description)},
    booktitle    = {Automated Deduction -- {CADE}-25},
    series       = {Lecture Notes in Computer Science},
    volume       = {9195},
    publisher    = {Springer},
    year         = {2015},
    pages        = {378--388}
}

@inproceedings{demoura2021lean4,
    author       = {de Moura, Leonardo and Ullrich, Sebastian},
    title        = {The Lean 4 Theorem Prover and Programming Language},
    booktitle    = {Automated Deduction -- {CADE} 28},
    series       = {Lecture Notes in Computer Science},
    volume       = {12699},
    publisher    = {Springer},
    year         = {2021},
    pages        = {625--635},
    doi          = {10.1007/978-3-030-79876-5_37},
    url          = {https://link.springer.com/chapter/10.1007/978-3-030-79876-5_37}
}
\end{document}